\documentclass{sig-alternate}

\newcommand{\eat}[1]{}

\usepackage{etex}
\reserveinserts{28}

\usepackage{latexsym}
\usepackage{amsfonts}
\usepackage{amsmath}
\usepackage{amssymb}
\usepackage{color}
\usepackage{colortbl}
\usepackage{epsfig}
\usepackage{xspace}
\usepackage{graphicx}
\usepackage{subfigure}
\usepackage{enumerate}
\usepackage[table]{xcolor}
\usepackage[all]{xy}
\usepackage{cite}
\usepackage{booktabs}
\usepackage{balance}

\usepackage{epsfig}
\usepackage{multirow}
\usepackage{url}

\usepackage[noend]{algorithmic}
\usepackage{algorithm}

\newcommand{\at}[1]{\protect\ensuremath{\mathsf{#1}}\xspace}

\newcommand{\stab}{\vspace{1.2ex}\noindent}

\newcommand{\ra}{\rightarrow}

\newcommand{\la}{\leftarrow}
\newcommand{\bi}{\begin{itemize}}
\newcommand{\ei}{\end{itemize}}
\newcommand{\mat}[2]{{\begin{tabbing}\hspace{#1}\=\+\kill #2\end{tabbing}}}

\newcommand{\be}{\begin{enumerate}}
\newcommand{\ee}{\end{enumerate}}
\newcommand{\beqn}{\begin{eqnarray*}}
\newcommand{\eeqn}{\end{eqnarray*}}

\newcommand{\stitle}[1]{\vspace{1ex}\noindent{\bf #1}}

\newcommand{\etitle}[1]{\vspace{0.8ex}\noindent{\underline{\em #1}}}

\newcommand{\ie}{\emph{i.e.,}\xspace}
\newcommand{\eg}{\emph{e.g.,}\xspace}

\newcommand{\True}{\mbox{\em true}}

\newcommand{\kw}[1]{{\ensuremath {\mathsf{#1}}}\xspace}

\newcounter{ccc}

\newcommand{\eop}{\hspace*{\fill}\mbox{$\Box$}\vspace{1ex}}     
\newcounter{example}
\renewcommand{\theexample}{\arabic{example}}
\newenvironment{example}{
        \vspace{1ex}
        \refstepcounter{example}
        {\noindent\bf Example \theexample:}}{
        \eop}

\renewcommand{\ni}{\noindent}

\newcommand{\nthesection}{\arabic{section}}
\newcounter{theorem}
\renewcommand{\thetheorem}{\arabic{theorem}}
\newcounter{prop}[section]

\newcounter{lemma}[section]
\renewcommand{\thelemma}{\nthesection.\arabic{theorem}}
\newcounter{cor}
\renewcommand{\thecor}{\arabic{theorem}}
\newenvironment{theorem}{\begin{em}
        \refstepcounter{theorem}
        {\vspace{1ex} \noindent\bf  Theorem  \thetheorem:}}{
        \end{em}\eop} 
\newenvironment{lemma}{\begin{em}
        \refstepcounter{theorem}
        {\vspace{1ex}\noindent\bf Lemma \thelemma:}}{
        \end{em}\eop} 
\newcounter{definition}[section]
\renewcommand{\thedefinition}{\nthesection.\arabic{definition}}

\newcounter{alg}[section]
\renewcommand{\thealg}{\nthesection.\arabic{alg}}

\newcounter{arule}
\renewcommand{\thearule}{\arabic{arule}}

\newcounter{claim}
\renewcommand{\theclaim}{\arabic{claim}}

\newcommand{\sys}{{\sf gLava}\xspace}

\newcommand{\cm}{{\sf CountMin}\xspace}
\newcommand{\gs}{{\sf gSketch}\xspace}

\newcommand{\cnt}{\kw{count}\xspace}

\newcommand{\weight}{\at{weight}\xspace}

\makeatletter
    \newcommand\figcaption{\def\@captype{figure}\caption}
    \newcommand\tabcaption{\def\@captype{table}\caption}
\makeatother

\usepackage{lipsum}   
\usepackage{framed}
\usepackage{tikz}
\usetikzlibrary{decorations.pathmorphing,calc}
\pgfmathsetseed{1} 
\pgfdeclarelayer{background}
\pgfsetlayers{background,main}

\tikzset{
  normal border/.style={orange!30!black!10, decorate, 
     decoration={random steps, segment length=2.5cm, amplitude=.7mm}},
  torn border/.style={orange!30!black!5, decorate, 
     decoration={random steps, segment length=.5cm, amplitude=1.7mm}}}

{\endMakeFramed}

\usepackage{etoolbox}
\makeatletter
\patchcmd{\maketitle}{\@copyrightspace}{}{}{}
\makeatother

\begin{document}

\title{On Summarizing Graph Streams}

\numberofauthors{1}

\author{ 
	Nan Tang \hspace{2ex} 
	Qing Chen\thanks{Qing is a QCRI intern, from Fudan University, China.} \hspace{2ex} 
	Prasenjit Mitra
	\\
	\affaddr{Qatar Computing Research Institute}
	\\
	\email{\{ntang, qchen, pmitra\}@qf.org.qa}
}

\maketitle
\date{}
\pagestyle{plain}

\begin{abstract}
Graph streams, which refer to the graph with edges being updated sequentially in a form of a stream, have wide applications such as cyber security, social networks and transportation networks. 
This paper studies the problem of {\em summarizing graph streams}. Specifically, given a graph stream $G$, directed or undirected, the objective is to summarize $G$ as $S_G$ with much smaller ({\em sublinear}) space, {\em linear} construction time and {\em constant} maintenance cost for each edge update, such that $S_G$ allows many queries over $G$ to be approximately conducted efficiently.
Due to the sheer volume and highly dynamic nature of graph streams, summarizing them remains a notoriously hard, if not impossible, problem.
The widely used practice of summarizing data streams is to treat each element independently by \eg hash- or sampling-based method, without keeping track of the connections between elements in a data stream, which gives these summaries limited power in supporting complicated queries over graph streams.
This paper discusses a fundamentally different philosophy for summarizing graph streams.
We present \sys, a probabilistic graph model that, instead of treating an edge (a stream element) as the operating unit, uses the finer grained node in an element.
This will naturally form a new graph sketch where edges capture the connections inside elements, and nodes maintain relationships across elements.
We discuss a wide range of supported graph queries and establish theoretical error bounds for basic queries.
%
\end{abstract}

\section{Introduction} 
\label{sec-intro}

Massive graphs arise in many applications \eg network traffic data, social networks, citation graphs, and transportation networks. These graphs  are highly dynamic: Network traffic data averages to about $10^9$ packets per hour per router for large ISPs~\cite{DBLP:journals/pvldb/GuhaM12}; Twitter sees 100 million users login daily, with around 500 millions tweets per day~\cite{twitterstat}.

\begin{example}
\label{exam:graph} 
We consider a graph stream\footnote{Without loss of generality, we use a directed graph for illustration. Our method can also work on undirected graphs.} as a sequence of elements $(x,y; t)$, which indicates that the edge $(x,y)$ is encountered at time $t$. 
A graph stream $\langle(a, b; t_1), (a, c; t_2), \cdots (b, a; t_{14})\rangle$ is depicted in Fig.~\ref{fig:graph}, where all timestamps are omitted for simplicity. 
Each edge is associated with a weight, which is $1$ by default.
\end{example}

Queries over big graph streams have many important applications, \eg monitoring cyber security attack and traffic networks. Unfortunately,
in a recent survey for businesses in analyzing big data streams~\cite{Vitria}, although 41\% of respondents stated that the ability to analyze and act on streaming data in minutes was critical, 67\% also admitted that they do not have the infrastructure to support that goal.
The situation is more complicated and challenging when handling graph streams.
It poses unique space and time constraints on storing, summarizing, maintaining and querying graph streams, due to its sheer volume, high velocity, and the complication of various queries. 

In fact, for the applications of data streams, fast and approximated answers are often preferred than exact answers.
Hence, sketch synopses have been widely studied for estimations over data streams.
Although they allow false positives, the space savings normally prevail over this drawback, when the probability of an error is sufficiently low.
There exist many sketches:~AMS~\cite{DBLP:journals/jcss/AlonMS99}, Lossy Counting~\cite{DBLP:conf/vldb/MankuM02}, \cm~\cite{DBLP:journals/jal/CormodeM05}, Bottom-$k$~\cite{DBLP:journals/pvldb/CohenK08} and \gs~\cite{DBLP:journals/pvldb/ZhaoAW11}. They have also been used in commercial tools, \eg~\cm is in Cisco OpenSOC\footnote{http://opensoc.github.io/} cyber security analytics framework.

\begin{figure}[t]
\begin{minipage}{\columnwidth}
\centerline{
\xymatrixcolsep{0.6in}
\xymatrix{
& b \ar@/^/[dl]^1 \ar[dd]^1 \ar[r]^1 \ar[ddr]^(.5)1 & d \ar[dr]^1  & \\
a \ar@/^/[ur]^1 \ar[dr]_1 & & e \ar[u]_(0.2)1 \ar[d]^1 \ar[ul]_(.3)1 & g \ar[llu]_(.4)1 \\
& c \ar[ur]^1 \ar[r]_1 & f \ar[llu]_(.7)1 & 
} 
} 
\caption{A sample graph stream}
\label{fig:graph}
\end{minipage}
\begin{minipage}{\columnwidth}
\vspace{4ex}
\centering
\begin{tabular}{|c|c|c|c|c|}
	\cline{2-5}
    \multicolumn{1}{c|}{$h(\cdot)$} & 5 & 5 & 3 & 1 \\ 
	\cline{2-5}
     \multicolumn{1}{c}{\color{white} {\Large I} } & \multicolumn{1}{c}{$ab, ac, ed, eb, ef$} & \multicolumn{1}{c}{$bc, bd, ba, bf, fa$} & \multicolumn{1}{c}{$ce, cf, gb$}  & \multicolumn{1}{c}{$dg$}\\
\end{tabular}
\caption{A hash-based sketch}
\label{fig:cm}
\end{minipage}
\begin{minipage}{\columnwidth}
\vspace{4ex}
\centerline{
\xymatrix{
& II (bf) \ar@/^/[dl]_2 \ar@/^1pc/[dd]^(.7){1} \ar@/_/[dr]^1 \ar@(ur,ul)[]^1 & \\
I (ae) \ar@/^/[ur]^3 \ar@/_/[dr]^1 \ar@/^/[rr]^1  & & III(cg) \ar@/^/[ll]_1 \ar@/_/[ul]_2 \\
& IV(d) \ar@/_/[ur]^1 & 
} 
} 
\caption{An example of our proposed sketch}
\label{fig:sys}
\end{minipage}
\end{figure}

\eat{
\begin{figure}[t]
\begin{minipage}{\columnwidth}
\centerline{
\xymatrix{
& 2 (b) \ar[dd]^{t_3} \ar[r]^{t_4} & 4(d) \ar[r]^{t_{10}} & 7(g) \ar[d]^{t_{11}} \\
1 (a) \ar[ur]^{t_1} \ar[dr]^{t_2} & & 5(e) \ar[u]^{t_7} \ar[d]^{t_8} \ar[r]^{t_9} & 8(b) \ar[ld]^{t_{12}} \ar[d]^{t_{14}}  \\
& 3(c) \ar[ur]^{t_5} \ar[r]^{t_6} & 6(f) \ar[r]^{t_{13}}  & 9(a)
} 
} 
\caption{An example graph stream}
\label{fig:graph}
\end{minipage}
\begin{minipage}{\columnwidth}
\vspace{4ex}
\centering
\begin{tabular}{|c|c|c|c|c|}
	\cline{2-5}
    \multicolumn{1}{c|}{$h(\cdot)$} & 5 & 5 & 3 & 1 \\ 
	\cline{2-5}
     \multicolumn{1}{c}{\color{white} {\Large I} } & \multicolumn{1}{c}{$ab, ac, ed, eb, ef$} & \multicolumn{1}{c}{$bc, bd, ba, bf, fa$} & \multicolumn{1}{c}{$ce, cf, gb$}  & \multicolumn{1}{c}{$dg$}\\
\end{tabular}
\caption{A hash-based sketch}
\label{fig:cm}
\end{minipage}
\begin{minipage}{\columnwidth}
\vspace{4ex}
\centerline{
\xymatrix{
& II (bf) \ar@/^/[dl]_2 \ar@/^1pc/[dd]^(.7){1} \ar@/_/[dr]^1 \ar@(ur,ul)[]^1 & \\
I (ae) \ar@/^/[ur]^3 \ar@/_/[dr]^1 \ar@/^/[rr]^1  & & III(cg) \ar@/^/[ll]_1 \ar@/_/[ul]_2 \\
& IV(d) \ar@/_/[ur]^1 & 
} 
} 
\caption{A sample of \sys sketch}
\label{fig:sys}
\end{minipage}
\end{figure}
} 

\begin{example}
\label{exam:gs}
Consider the graph stream in Example~\ref{exam:graph} and Fig.~\ref{fig:graph}.
\cm treats each element in the stream independently, and maps them to $w$ hash buckets, by using the pairs of node labels as hash keys. 
Assume there are $w=4$ hash buckets. 
The set of values to be hashed is $\{ab, ac, bc, \cdots\}$, where the value $ab$ is to concatenate two node labels for the element $(a, b;t_1)$.
A hash function is used to put these values in 4 buckets, \ie $h(\cdot)\ra [1, 4]$, as shown in Fig.~\ref{fig:cm}.
The first bucket aggregates the weights of values starting with $a$ and $e$, and similar for other buckets.
Note that, for simplicity, we use only one hash function for illustration. In fact, \cm uses $d$ pairwise independent hash functions to alleviate the problem of hash key collisions.

\cm, and its variant \gs, can be used for estimating certain types of graph queries. 
(a) {\em Edge query}. It can estimate the weight of particular edges: $5$ for $ab$, and $1$ for $dg$.
(b) {\em Aggregate subgraph query.} It can answer queries like: {\em What is the aggregated weight for a graph with two edges $(a, c)$ and $(c, e)$}? It will give $5+3=8$ as the estimated weight for the graph. 
\end{example}

Example~\ref{exam:gs} shows the applicability of \cm on some query estimation over graph streams.
However, their (and other sketch synopses') main weakness is that they are element (or edge) independent. That is, they fall short of maintaining the connections between streaming elements, thus fail in estimating many important queries, such as node monitoring, path queries (\eg reachability) and more complicated graph analytics (see Section~\ref{subsec-queries} for a detailed discussion).

\stitle{Challenges.}
Designing a graph sketch to support a wide range of applications requires to satisfy the following constraints.
(1) {\em Space constraint:} {\em sublinear} upper bound is needed.
(2) {\em Time constraint:} {\em linear} construction time is required. Noticeably, this is stronger than the constraint with only a constant passes over the stream.
(3) {\em Maintenance constraint:} to maintain it for one element insertion/deletion should be in {\em constant} time.
(4) {\em Element connectivity:} the {\em connectivity} between elements should be maintained.

To this end, we present \sys, a novel generalized graph sketch to meet the above required constraints. Instead of treating elements (edges) in a graph stream independently as those prior art, the key idea of \sys is to compress a graph stream based on a finer grained item, the node in a stream element. 

\begin{example}
\label{exam:sys}
Again, consider the graph stream in Fig.~\ref{fig:graph}. Our proposed sketch is shown in Fig.~\ref{fig:sys}.
For each edge $(x, y; t)$, \sys uses a hash function to map each node label to 4 node buckets \ie~$h'(\cdot)\ra[1,4]$. 
Node $I$ is the summary of two node labels $a$ and $e$, assuming $h'(a) = 1$ and $h'(e) =1$. The other compressed nodes are computed similarly.
The edge weight denotes the aggregated weights from stream elements, \eg the number $3$ from node $I$ to $II$ means that there are three elements as $(x, y; t)$ where the label of $x$ (resp. $y$) is $a$ or $e$ (resp. $b$ or $f$).
\end{example}

\etitle{Remark.}
(1) It is readily to see that estimation for edge frequencies, as what \cm supports, can be easily achieved by \sys. Better still, the idea of using multiple hash functions to alleviate hash collisions can be easily applied to \sys.
(2) \sys is represented as a graph, which captures not only the connections inside elements, but also the links across elements.
These make it an ideal sketch to support a much wider range of applications, compared with those prior art (see Section~\ref{subsec-queries} for a discussion).

\stitle{Contributions.}
This paper presents a novel graph sketch for supporting a wide range of graph stream applications.

\stab(1)
We introduce \sys, a novel graph sketch (Section~\ref{subsec-gmodel}).
As shown in Fig.~\ref{fig:sys}, the proposed sketch naturally preserves the graphical connections of the original graph stream, which makes it a better fit than traditional data stream sketches to support analytics over graph streams. 
We further categorize its supported graph queries (Section~\ref{subsec-queries}).

\stab(2)
We describe algorithms to process various graph analytics on top of \sys (Section~\ref{sec-algs}).
The general purpose is, instead of proposing new algorithms, to show that \sys can easily be used to support many graph queries, and {\em off-the-shelf} graph algorithms can be seamlessly integrated.

\stab(3) We perform theoretical analysis to establish error bounds for basic queries, specifically, the edge frequency estimation and point query estimation (Section~\ref{sec-theory}).

\stab(4)
We describe implementation details (Section~\ref{sec-implementation}).
This is important to ensure that the new sketch can be constructed and maintained under the hard time/space constraints to support streaming applications.
Moreover, we propose to use non-square matrices, using the same space, to improve the accuracy of estimation.




\vspace{1ex}

In this history of graph problems over streams, unfortunately, most results showed that a large amount of space is required for {\em complicated} graph problems~\cite{DBLP:books/crc/aggarwal2014, DBLP:journals/sigmod/McGregor14}.
The present study fills a gap in the literature by analyzing various graph problems in a small amount of space.
We thus contend that \sys will shed new light on graph stream management.

\stitle{Organization.}
Section~\ref{sec-related} discusses related work.
Section~\ref{sec-gapollo} defines the new sketch and queries to be solved.
Section~\ref{sec-algs} describes algorithms using \sys.
Section~\ref{sec-theory} makes theoretical analysis.
Section~\ref{sec-implementation} describes implementation details.
Finally, Section~\ref{sec-conclusion} concludes the paper with a summary of our findings.

\section{Related Work} 
\label{sec-related}

We categorize related work as follows.

\etitle{Sketch synopses.}
Given a data stream, the aim of sketch synopses is to apply (linear) projections of the data into lower dimensional spaces that preserve the salient features of the data. A considerable amount of literature has been published on general data streams such as AMS~\cite{DBLP:journals/jcss/AlonMS99}, lossy counting~\cite{DBLP:conf/vldb/MankuM02}, \cm~\cite{DBLP:journals/jal/CormodeM05} and bottom-$k$~\cite{DBLP:journals/pvldb/CohenK08}. The work \gs~\cite{DBLP:journals/pvldb/ZhaoAW11} improves \cm for graph streams, by assuming that data samples or query samples are given.
Bloom filters have been widely used in a variety of network problems (see~\cite{DBLP:journals/im/BroderM03} for a survey).
There are also sketches that maintain counters only for nodes, \eg using a heap for maintaining the nodes the largest degrees for heavy hitters~\cite{DBLP:conf/pods/CormodeM05}. 

As remarked earlier, \sys is more general, since it maintains connections inside and across elements, {\em without} assuming  any sample data or query is given. None of existing sketches maintains both node and edge information.

\etitle{Graph summaries.}
Summarizing graphs has been widely studied. The most prolific area is in web graph compression. The papers~\cite{DBLP:conf/dcc/AdlerM01,DBLP:conf/dcc/SuelY01} encode Web pages with similar adjacency lists using reference encoding, so as to reduce the number of bits needed to encode a link. The work~\cite{DBLP:conf/icde/RaghavanG03} groups Web pages based on a combination of their URL patterns and $k$-means clustering.
The paper~\cite{DBLP:conf/sigmod/FanLWW12} compresses graphs based on specific types of queries.
There are also many clustering based methods from data mining community (see \eg~\cite{DBLP:books/ph/JainD88}), with the basic idea to group {\em similar} nodes together.

These data structures are designed for less dynamic graphs, which are not suitable for graph stream scenarios.

\etitle{Graph pattern matching over streams.}
There have been several work on matching graph patterns over graph streams, based on either the semantics of subgraph isomorphism~\cite{DBLP:conf/icde/WangC09, DBLP:conf/edbt/ChoudhuryHCAF15, DBLP:conf/icde/GaoZZY14a} or graph simulation~\cite{DBLP:journals/pvldb/SongGCW14}. 
The work~\cite{DBLP:conf/icde/WangC09}  assumes that queries are given, and builds node-neighbor tree to filter false candidate results. 
The paper~\cite{DBLP:conf/icde/GaoZZY14a} leverages a distributed graph processing framework, Giraph, to approximately evaluate graph quereis.
The work~\cite{DBLP:conf/edbt/ChoudhuryHCAF15} uses the subgraph distributional statistics collected from the graph streams to optimize a graph query evaluation.
The paper~\cite{DBLP:journals/pvldb/SongGCW14} uses filtering methods to find data that potentially matches for a specific type of queries, namely {\em degree-preserving dual simulation with timing constraints}. 

Firstly, all the above algorithms are designed for a certain type of graph queries.
Secondly, most of them assume the presence of queries, so they can build indices to accelerate.
In contrast, \sys aims at summarizing graph streams in a generalized way, so as to support various types of queries, {\em without} any assumption of queries.

\etitle{Graph stream algorithms.}
There has also been work on algorithms over graph streams (see~\cite{DBLP:journals/sigmod/McGregor14} for a survey). This includes the problems of connectivity~\cite{DBLP:journals/tcs/FeigenbaumKMSZ05}, trees~\cite{targan1983}, spanners~\cite{DBLP:journals/talg/Elkin11}, sparsification~\cite{DBLP:journals/mst/KelnerL13}, counting subgraphs \eg triangles~\cite{DBLP:conf/icalp/BravermanOV13, DBLP:conf/kdd/TsourakakisKMF09}. However, they mainly focus on theoretical study for best approximation bound, mostly on $O(n$ polylog $n)$ space, with one to multiple passes over the data stream.

In fact, \sys is a friend, instead a competitor, of them. As will be seen later (Section~\ref{sec-algs}), \sys can treat existing algorithms as {\em black-boxes} to help solve existing problems.

\etitle{Distributed graph systems.}
Many distributed graph computing systems have been proposed to conduct data processing and data analytics in massive graphs, such as Pregel~\cite{DBLP:conf/sigmod/MalewiczABDHLC10}, Giraph\footnote{http://giraph.apache.org}, GraphLab~\cite{DBLP:journals/pvldb/LowGKBGH12}, Power-Graph~\cite{DBLP:conf/osdi/GonzalezLGBG12} and GraphX~\cite{DBLP:conf/osdi/GonzalezXDCFS14}.

They have been proved to be efficient on static graphs, but are not ready for doing analytics over big graph streams with real-time response. In fact, they are complementary to and can be used for \sys in distributed settings.

\section{gLava and Supported Queries} 
\label{sec-gapollo}

We first define graph streams (Section~\ref{subsec-gs}) and state the studied problem (Section~\ref{subsec-ps}).
We then introduce our proposed graph sketch model (Section~\ref{subsec-gmodel}).
Finally, we discuss the queries supported by our new sketch (Section~\ref{subsec-queries}).

\subsection{Graph Streams}
\label{subsec-gs}

A {\em graph stream} is a sequence of elements $e = (x, y; t)$ where $x, y$ are node identifiers (labels) and edge $(x, y)$ is encountered at time-stamp $t$. Such a stream,
\[
G=\langle e_1, e_2, \cdots, e_m \rangle
\]

\ni naturally defines a graph $G=(V, E)$ where $V$ is a set of nodes and $E=\{e_1, \cdots, e_m\}$.
We write $\omega(e_i)$ the weight for the edge $e_i$, and  
$\omega(x, y)$ the aggregated edge weight from node $x$ to node $y$.
We call $m$ the {\em size} of the graph stream, denoted by $|G|$.

Intuitively, the node label, being treated as an identifier, uniquely identifies a node, which could be \eg IP addresses in network traffic data or user IDs in social networks. 
Note that, in the graph terminology, a graph stream is a {\em multigraph}, where each edge can occur many times, \eg one IP address can send multiple packets to another IP address. We are interested in properties of the underlying graph.
The causes a main challenge with graph streams where one normally does not have enough space to record the edges that have been seen so far. 
Summarizing such graph streams in one pass is important to many applications.
%
%
Moreover, we do not explicitly differentiate whether $(x, y)$ is an ordered pair or not. In other words, our approach applies naturally to both directed and undirected graphs. 
 
For instance, in network traffic data, a stream element arrives at the form: $(192.168.29.1, 192.168.29.133, 62, 105.12)$\footnote{We omit port numbers and protocols for simplicity.}, where node labels $192.168.29.1$ and $192.168.29.133$ are IP addresses, $62$ is the number of bytes sent {\em from} $192.168.29.1$ {\em to} $192.168.29.133$ in this captured packet (\ie the weight of the directed edge), and $105.12$ is the time in seconds that this edge arrived when the server started to capture data. Please see Fig.~\ref{fig:graph} for a sample graph stream.
%
%

\subsection{Problem statement}
\label{subsec-ps}

The problem of {\em summarizing graph streams} is, given a graph stream $G$, to design another data structure $S_G$ from $G$, such that:

\be
	\item $|S_G| \ll |G|$: the size of $S_G$ is far less than $G$, preferably in sublinear space.
	\item The time to construct $S_G$ from $G$ is in linear time.
	\item The update cost of $S_G$ for each edge insertion/deletion is in constant time.
\ee

Intuitively, a graph stream summary has to be built and maintained in real time, so as to deal with big volume and high velocity graph stream scenarios.
In fact, the \cm~\cite{DBLP:conf/vldb/MankuM02} and its variant \gs~\cite{DBLP:journals/pvldb/ZhaoAW11} satisfy the above conditions (see Example~\ref{exam:gs} for more details). Unfortunately, as discussed earlier, \cm and \gs can support only limited types of graph queries.

\subsection{The gLava Model}
\label{subsec-gmodel}

\stitle{The graph sketch.}
A  {\em graph sketch} is a graph $S_G({\cal V}, {\cal E})$, where ${\cal V}$ denotes the set of vertices and ${\cal E}$ its edges. For vertex $v\in {\cal V}$, we simply treat its label as its node identifier (the same as the graph stream model).
Each edge $e$ is associated with a weight, denoted as $\omega(e)$. 

 In generating the above graph sketch $S_G$ from a graph $G$, we first set the number of nodes in the sketch, \ie~let $|{\cal V}| = w$. For an edge $(x, y;t)$ in $G$, we use a hash function $h$ to map the label of each node to a value in $[1, w]$, and the aggregated edge weight is calculated correspondingly.

Please refer to Fig.~\ref{fig:sys} as an example, where we set $w=4$.

\etitle{Discussion about edge weight.}
The edge weight for an edge $e$ in the graph sketch is computed by an {\em aggregation} function of all edge weights that are mapped to $e$.
Such an aggregation function could be $\at{min}(\cdot)$, $\at{max}(\cdot)$, $\at{count}(\cdot)$, $\at{average}(\cdot)$, $\at{sum}(\cdot)$ or other functions.
In this paper, we use $\at{sum}(\cdot)$ by default to explain our method. 
The other aggregation functions can be similarly applied.
In practice, which aggregation function to use is determined by applications.
For a more general setting that requires to maintain multiple aggregated functions in a graph sketch, 
we may extend our model to have multiple edge weights \eg~$\omega_1(e), \cdots ,\omega_n(e)$, 
with each $\omega_i(e)$ corresponds to an aggregation function.

\etitle{Remark.}
One may observe that the graph sketch model is basically the same as the graph stream model, with the main difference that the time-stamps are not maintained. This makes it very useful and applicable in many scenarios when querying a graph stream for a given time window. In other words, for a graph analytics method $M$ that needs to run on a streaming graph $G$, denoted as $M(G)$, one can run it on its sketch $S_G$, \ie~$M(S_G)$, to get an estimated result, without modifying the method $M$.

\begin{example}
\label{exam:gsedge}
Consider the graph stream in Fig~\ref{fig:graph} and its sketch in Fig.~\ref{fig:sys}. 
Assume that query $Q_1$ is to estimate the aggregated weight of edges from $b$ to $c$.
In Fig.~\ref{fig:sys}, one can map $b$ to node $II$, $c$ to node $III$, and get the estimated weight $1$ from edge $(II, III)$, which is precise.
Now consider  $Q_2$, which is to compute the aggregated weight from $g$ to $b$.
One can locate the edge $(III, II)$, and the estimated result is $2$, which is not accurate since the real weight of $(g, b)$ in Fig~\ref{fig:graph} is $1$.
\end{example}
 
The above result is expected, since given the compression, no hash function can ensure that the estimation on the sketch can be done precisely. Along the same line of \cm~\cite{DBLP:journals/jal/CormodeM05}, we use multiple independent hash functions to reduce the probability of hash collisions.

\begin{figure}[t]
\hspace*{2ex}
\begin{minipage}{0.4\columnwidth}
\centerline{
\xymatrixcolsep{0.1in}
\xymatrix{
& II (bc) \ar@/^/[dl]_3 \ar@/^1pc/[dd]^(.7){1} \ar@/_/[dr]^1 \ar@(ur,ul)[]^1 & \\
I (af) \ar@(u,l)[]^1 \ar@/^/[ur]^2  & & III(dg) \ar@/_/[ul]_1 \ar@(u,r)[]^1 \\
& IV(e) \ar@/_/[ur]^1 \ar[ul]^1 \ar@/^/[uu]^(.3){1} & \\
} 
} 
\centerline{
\xymatrixrowsep{0.2in}
\xymatrix{
& \txt{(a) Sketch $S_1$} &
} 
} 
\end{minipage}
\hspace*{6ex}
\begin{minipage}{0.4\columnwidth}
\centerline{
\xymatrixcolsep{0.1in}
\xymatrix{
& ii (cd) \ar@/^1pc/[dd]^(.4){2} \ar[dr]^1 & \\
i (ab) \ar[ur]^3 \ar@(ru,lu)[]^2  \ar@/_/[dr]_1 & & iii(g) \ar@/^/[ll]_1 \\
& iv(ef) \ar@/_/[ul]^2 \ar@(dr,dl)[]_1 \ar@/^/[uu]^(.6){1} & 
} 
} 
\centerline{
\xymatrixrowsep{0.2in}
\xymatrix{
& \txt{(b) Sketch $S_2$} &
} 
} 
\end{minipage}
\caption{A \sys sketch with 2 hash functions}
\label{fig:2hash}
\end{figure}
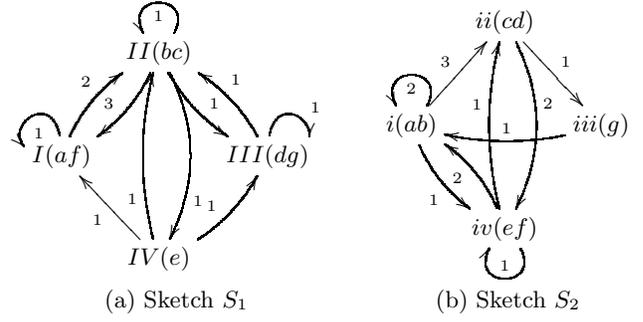

\stitle{The \sys model.}
A \sys sketch is a set of graph sketches $\{S_1({\cal V}_1, {\cal E}_1), \cdots, S_d({\cal V}_d, {\cal E}_d)\}$. Here, we use $d$ hash functions $h_1, \cdots, h_d$, where $h_i$ ($i\in[1, d]$) is used to generate $S_i$. Also, $h_1, \cdots, h_d$ are chosen uniformly at random from a pairwise-independent family (see Section~\ref{subsec-pairwise} for more details).
 
\begin{example}
\label{exam:2hash}
Figure~\ref{fig:2hash} shows another two sketches for Fig.~\ref{fig:graph}. 
Again, consider the query $Q_2$ in Example~\ref{exam:gsedge}.
Using $S_1$ in Fig.~\ref{fig:2hash} (a), one can locate edge $(III, II)$ for $(g, b)$, which gives $1$.
Similarly, $S_2$ in Fig.~\ref{fig:2hash} (b) will also give $1$ from edge $(iii, i)$, where $g$ (resp. $b$) maps to $iii$ (resp. $i$). The minimum of the above two outputs is $1$, which is correct. 
\end{example}

Example~\ref{exam:2hash} shows that using multiple hash functions can indeed improve the accuracy of estimation.

\subsection{Supported Graph Queries}
\label{subsec-queries}

As remarked in Section~\ref{subsec-gmodel}, for any graph analytics method $M$ to run over $G$, \ie~$M(G)$, it is possible to run $M$ on each sketch directly and individually, and then merge the result as: $\tilde{M}(G) = \Gamma(M(S_1), \cdots, M(S_d))$, where $\tilde{M}(G)$ denotes the estimated result over $G$, and $\Gamma(\cdot)$ an aggregation function (\eg~\at{min}, \at{max}, \at{conjunction}) to merge results returned from $d$ sketches.

Whilst the exercise in this section only substantiates several types of graph queries to be addressed in this work,
it is evident that the power of \sys is far beyond what are listed below. 

\stitle{Edge frequency.}
Given two node labels $a$ and $b$, we write $f_e(a, b)$ to denote the exact weight from a $a$-node to a $b$-node.
We write $\tilde{f_e}(a, b)$ the estimated weight from a sketch.


One application of such queries, taking social networks for example, is to estimate the communication frequency between two specific friends.

\stitle{Point queries.}
For a directed graph, given a node label $a$, we study a boolean query $f_v(a, \la) > \theta$ (resp. $f_v(a, \ra) < \theta$), which is to monitor whether the aggregated edge weight {\em to} (resp. {\em from}) a node with label $a$ is above (resp. below) a given threshold $\theta$ in the graph stream $G$. For an undirected graph, we write $f_v(a, \perp) > \theta$ and $f_v(a, \perp) < \theta$ correspondingly. 
Similarly, we use $\tilde{f_v}(a, \la), \tilde{f_v}(a, \ra), \tilde{f_v}(a, \perp)$ for estimated numbers using sketches.


One important application of such queries is DoS (Denial-of-service) attacks in cyber security, which typically flood a target source (\ie a computer) with massive external communication requests. 

\stitle{Path queries.}
Given two node labels $a$ and $b$, a (boolean) reachability query $r(a, b)$ is to tell whether they are connected. Also, we write $\tilde{r}(a, b)$ as the estimated reachability.

One important monitoring task, for the success of multicast deployment in the Internet, is to verify the availability of service in the network, which is usually referred to as {\em reachability monitoring}. 
Another reachability application that needs to consider edge weights is IP routing, which is to determine the path of data flows in order to travel across multiple networks.

\stitle{Aggregate subgraph query.}
%
The {\em aggregate subgraph query} is considered in~\gs~\cite{DBLP:journals/pvldb/ZhaoAW11}.
It is to compute the aggregate weight of the constituent edges of a sub-graph 
$Q = \{(x_1,y_1), \cdots, (x_k,y_k)\}$, denoted by $f(Q) = \Omega (f_e(x_1, y_1), \cdots, f_e(x_k,y_k))$. 
Here, the function $\Omega(\cdot)$ is to merge weights from all $f_e(x_i, y_i)$ for $i\in[1, k]$.
We write $\tilde{f}(Q)$ for estimation over sketches.

Note that, \gs simply merges all estimated weights in $G$, \eg the sum of all weights, even if one edge is missing (\ie~$\tilde{f}(x_i, y_i)=0$ for some edge $e_i$). 
However, we use a different query semantics in this work: if $\tilde{f}(x_i, y_i)=0$, the estimated aggregate weight should be $0$, since the query graph $Q$ does not have an exact match.
We use this revised semantics since it is a more practical setting.

\begin{example}
\label{exam:subgraph}
Consider a subgraph with two edges as $Q: \{(a, b), (a, c)\}$.
The query $Q_3: \tilde{f}(Q)$ is to estimate the aggregate weight of $Q$.
The precise answer is $2$, which is easy to check from Fig.~\ref{fig:graph}.
\end{example}

\etitle{Extensions.}
We consider an extension of the above aggregate subgraph query, which allows a wildcard $*$ in the node labels that are being queried.
More specifically, for the subgraph query $Q = \{(x_1,y_1), \cdots, (x_k,y_k)\}$, each $x_i$ or $y_i$ is either a constant value, or a wildcard $*$ (\ie match any label).
A further extension is to {\em bound} the wildcards to be matched to the same node, by using a subscript to a wildcard as $*_j$. That is, two wildcards with the same subscripts enforce them to be mapped to the same node.

\begin{example}
\label{exam:subgraph}
A subgraph query $Q_4: \tilde{f}(\{(a, b), (b, c), (c, a)\})$ is to estimate the triangle, \ie a 3-clique with three vertices labeled as $a$, $b$, and $c$, respectively.
Another subgraph query $Q_5: \tilde{f}(\{(*, b), (b, c), (c, *)\})$ is to estimate paths that start at node with an edge to $b$, and end at any node with an edge from $c$, if the edge $(b, c)$ exists in the graph.
If one wants to count the common neighbors of $(b, c)$, the following query $Q_6: \tilde{f}(\{(*_1, b), (b, c), (c, *_1)\})$ can enforce such a condition, which again is a case of counting triangles.
\end{example}

The extension is apparently more general, with the purpose to cover more useful queries in practice.
Unfortunately, \gs cannot support such extensions.


\begin{table}[t]
\begin{center}
{
\begin{tabular}{| c | c | }
\hline
    \cellcolor{black}{\textcolor{white}{\at{Symbols}}} &
    \cellcolor{black}{\textcolor{white}{\at{Notations}}} \\ \hline
	$G$, $S_G$ & a graph stream, and a graph sketch \\ \hline
	$\omega(e)$ & weight of the edge $e$ \\ \hline
	$f_e(a,b)$ & edge weight \\ \hline
	$f_v(a, \la)$ & node in-flow weight (directed graphs) \\ \hline
	$f_v(a, \ra)$ & node out-flow weight (directed graphs) \\ \hline
	$f_v(a, \perp)$ & node flow weight (undirected graphs) \\ \hline	
	$r(a, b)$ & whether $b$ is reachable from $a$ \\ \hline
	$f(Q)$ & weight of subgraph $Q$ \\ \hline
\end{tabular}
} 
\end{center}
\vspace{-2.5ex}
\caption{Notations}\label{tbl:notation}
\end{table}


\etitle{Summary of notations.} 
The notations of this paper are summarized in Table~\ref{tbl:notation}, which are for both directed and undirected graphs, unless being specified otherwise.

\section{Query Processing}
\label{sec-algs}

We consider the graph stream as $G$, and $d$ graph sketches $\{S_1, \cdots, S_d\}$ of $G$, where $S_i$ is constructed using a hash function $h_i(\cdot)$.
In this following, we discuss how to process different types of queries. 
Again, by default, we assume that we use $\at{sum}(\cdot)$ as the default aggregation function.
Also, to facilitate the discussion, we assume that an adjacent matrix is used for storing \sys.

\subsection{Edge Query}
\label{subsec:edgeq}

To evaluate $\tilde{f_e}(a,b)$, the edge weight of $(a, b)$, is straightforward.
It is to find the estimated edge weight from each sketch $\omega_i(h_i(a), h_i(b))$ and then use a corresponding function $\Gamma(\cdot)$ to merge them, as the following:
\[
\tilde{f_e}(a,b) = \Gamma(\omega_1(h_1(a), h_1(b)), \cdots, \omega_d(h_d(a), h_d(b)))
\]

In this case, the function $\Gamma(\cdot)$ is to take the minimum.

\stitle{Complexity.}
It is easy to see that estimating the aggregate weight of an edge query is in $O(d)$ time and in $O(d)$ space, where $d$ is a constant.

\subsection{Point Queries}
\label{subsec:nodeq}

Here, we only discuss the case $\tilde{f_v}(a, \la) > \theta$. That is, given a new element $(x, y; t)$, estimate in real-time whether the aggregate weight to node $a$ is above the threshold $\theta$.
The other cases, \ie~$\tilde{f_v}(a, \la) < \theta$, $\tilde{f_v}(a, \ra) > \theta$, $\tilde{f_v}(a, \ra) < \theta$, $\tilde{f_v}(a, \perp) > \theta$ and $\tilde{f_v}(a, \perp) < \theta$, can be similarly processed.

We use the following strategy to monitor $\tilde{f_v}(a, \la) > \theta$, given an incoming edge $e: (x, y; t)$.
Note that if $y \ne a$, we simply update all $d$ sketches in constant time, since this is not about an edge to $a$. Next, we only describe the case of edge $e: (x, a; t)$.

\etitle{Step 1.} [Estimate current in-flow.]
We write $\tilde{f_v^i}(a, \la)$ the estimated in-flow from the $i$-th sketch.
$\tilde{f_v^i}(a, \la)$ can be computed by first locating the column in the adjacent matrix corresponding to label $a$ (\ie~$h_i(a)$), and then sum up the values in that column, \ie $\tilde{f_v^i}(a, \la) = \sum_{j=1}^w {\cal M}_i[j][h_i(a)]$. Here, $w$ is the width of the adjacent matrix ${\cal M}_i$. Then, 
\[
\tilde{f_v}(a, \la) = \Gamma(\tilde{f_v^1}(a, \la), \cdots, \tilde{f_v^w}(a, \la))
\]

In this case, the function $\Gamma(\cdot)$ is to take the minimum.

\etitle{Step 2.} [Monitor $e$.]
We only send an alarm, if $\tilde{f_v}(a, \la) + \omega(e) > \theta$.

\etitle{Step 3.} [Update all sketches.]
Update all $d$ sketches by aggregating the new edge weight $\omega(e)$.

\stitle{Complexity.}
It is easy to see that step~1 takes $O(d+w)$ time and $O(d)$ space, where both $d$ and $w$ are constants.

\subsection{Path Queries}
\label{subsec:pathq}

We consider the reachability query $\tilde{r}(a, b)$, which is to estimate whether $b$ is reachable from $a$. For such queries, we treat any off-the-shelf algorithm $\at{reach}(x, y)$ as a {\em black-box} and show our strategy.

\etitle{Step 1.} [Map.]
We invoke $\at{reach}_i(h_i[a], h_i[b])$ on the $i$-th sketch (for $i\in[1, d]$), to decide whether the mapped node $h_i[b]$ is reachable from the mapped node $h_i[a]$.

\etitle{Step 2.} [Reduce.]
We merge individual results as follows:
\[
	\tilde{r}(a, b) = \at{reach}_1(h_1[a], h_1[b]) \wedge \cdots \wedge \at{reach}_d(h_d[a], h_d[b])
\] 
That is, the estimated result is \True~only if the mapped nodes are reachable from {\em all} $d$ sketches.

The complexity of the above strategy is determined by the algorithm $\at{reach}()$.

\subsection{Aggregate Subgraph Query}
\label{subsec:subgraphq}

We next consider the aggregate subgraph query $\tilde{f}(Q)$, which is to compute the aggregate weight of the constituent edges of a sub-graph $Q$.
The process is similar to the above path queries, by using any existing algorithm $\at{subgraph}(Q)$.

\etitle{Step 1.} [Map.]
We first invoke $\at{subgraph}(Q)$ at each sketch to find subgraph matches, and calculate the aggregate weight, denoted by $\weight_i(Q)$ for the $i$-th sketch.

\etitle{Step 2.} [Reduce.]
We merge individual results as follows:
\[
	\tilde{f}(Q) = \at{min}(\weight_1(Q), \cdots, \weight_d(Q))
\] 

Note that, running a graph algorithm on a sketch is only applicable to \sys. It is not applicable to \gs since \gs by nature is an array of frequency counting, without maintaining the graphical structure as \sys does.
Also, in the case $\weight_i(Q)$ from some sketch that says that a subgraph match does not exist, we can terminate the whole process, which provides chances for optimization.


\stitle{Optimization.}
Recall that a subgraph $Q$ is defined as constituent edges as $\{(x_1,y_1), \cdots, (x_k,y_k)\}$ (Section~\ref{subsec-queries}).
An alternative way of estimating the aggregate subgraph query is to first compute the minimum value of each value and sum them up.
We denote this approach by $\tilde{f'}(Q)$ as follows:
\[
	\tilde{f'}(Q) = \sum_{i=1}^k \tilde{f_e}(x_i, y_i) 
\] 

It is readily to see that $\tilde{f'}(Q) \leq \tilde{f}(Q)$.
Recall that there are two extensions of the subgraph queries (Section~\ref{subsec-queries}).
For the first extension that a wildcard $*$ is used, the above optimization can be used. For instance, the edge frequency of $\tilde{f_e}(x, *)$ for a directed graph is indeed $\tilde{f_v}(x, \ra)$.
For the second extension that multiple wildcards $*_i$ are used to bound to the same node, this optimization cannot be applied.

\section{Error Bounds}
\label{sec-theory}

We study the error bounds for two basic types of queries: edge queries and point queries.
Although \sys is more general, we show that theoretically, it has the same error bounds 
as \cm.

\subsection{Edge Frequency} 

Our proof for the error bound of edge frequency queries is an adaption of the proof used in 
\cm in Sec.~4.1 of \cite{DBLP:journals/jal/CormodeM05}. 

\begin{theorem}
The estimation $\tilde{f_e}(j, k)$ of the cumulative edge weight of the edge $(j,k)$  has the following guarantees, $f_e(j, k) \leq \tilde{f_e}(j, k)$ with probability at least $1 - \delta$, s.t.,
$f_e(j, k) \le \tilde{f_e}(j, k) + \epsilon * n$,
where $f_e(j, k)$ is the exact answer to the cumulative edge weight of the edge $(j, k)$, 
 $n$ denotes the number of nodes,
and the error in answering a query is within a factor of $\epsilon$ with probability $\delta$\footnote{The parameters $\epsilon$ and $\delta$ are usually set by the user.}.
\end{theorem}

\begin{proof}
For any query edge $f_e(j, k)$, for each hash function $h_i$, by construction,
for any stream edge $e: (j, k; t)$, where $\omega(e)$ is the edge weight that is added to 
$\cnt(h_i(j),h_i(k))$.  Therefore, by construction, the answer to
$f_e(j, k)$ is less than or equal to $min_i (\cnt(h_i(j),h_i(k)))$ for $i\in[1, d]$ where $d$ is the number of hash functions.  

Consider two edges $(j, k)$ and $(l, m)$.
We define an indicator variable $I_{i,j,k,l,m}$, which is
1 if there is a collision between two distinct edges, \ie $(j \neq l) \wedge (k \neq m) \wedge (h_i(j)=h_i(l) \wedge h_i(k)=h_i(l))$ , and 0 otherwise.  
By pairwise independence of the hash functions,

\mat{4ex}{
$E(I_{i,j,k,l,m})$ \= $= Pr[(h_i(j)=h_i(l)  \wedge h_i(k)=h_i(m)]$ \\
\> $\leq (1/\at{range}(h_i))^2=\epsilon'^2/e^2$ \textcolor{white}{\huge P}
}

\ni where $e$ is the base used to compute natural logarithm, and $\at{range}(h_i)$ is the
number of hash buckets of function $h_i$.
Define the variable $X_{i,j,k}$ (random over choices of $h_i$) to count the number of collisions with the edge $(j, k)$,
which is formalized as
$X_{i,j,k}=\Sigma_{l=1 \ldots n,m=1 \ldots n} I_{i,j,k,l,m} a_{l,m}$, where $n$ is the number of nodes in the graph.
Since all $a_i$'s are non-negative in this case, $X_{i,j,k}$ is non-negative.
We write $\cnt[i,h_i(j),h_i(k)]$ as the count in the hash bucket relative to hash function $h_i$.
By construction, $\cnt[i,h_i(j),h_i(k)] = f_e(j,k) + X_{i,j,k}$.  
So clearly, $min_i (\cnt[i,h_i(j),h_i(k)]) \geq f_e(j, k)$.

\mat{4ex}{
	\=$E(X_{i,j,k})$   \\
	\>\hspace{4ex}\=$=E(\sum_{l=1 \ldots n, m=1 \ldots n}I_{i,j,k,l}f_e(l,m))$  \textcolor{white}{\huge P} \\
	\>\>$\leq (\sum_{l=1 \ldots n, m=1 \ldots n}E(I_{i,j,k,l}).f_e(l,m)) \leq (\epsilon'/e)^2* n$ \textcolor{white}{\huge P}
}

\ni by pairwise independence of $h_i$, and linearlity of expectation.

Let, $\epsilon'^2 = \epsilon$.  By the Markov inequality,

\mat{4ex}{
	\= $Pr[\tilde{f_e}{(j,k)} > f_e(j,k) + \epsilon * n]$ \= \\
	\>\hspace{4ex}\= $=Pr[\forall_i. \at{count}[i, h_i(j), h_i(k)] > f_e(j,k) + \epsilon * n$] \textcolor{white}{\huge P}\\
 	\>\> $=Pr[\forall_i. f_e(j,k)+X_{i,j,k}>f_e(j,k)+ \epsilon * n]$ \textcolor{white}{\huge P}\\
	\>\> $=Pr[\forall_i. X_{i,j,k} > e * E(X_{i,j,k})] < e^{-d} \leq \delta$  \textcolor{white}{\huge P}
}
\end{proof}

Our algorithm generates the same number of collisions and the same error bounds under the same probabilistic guarantees as \cm for edge frequency queries.

\subsection{Point Queries}

We first discuss the query for node out-degree, \ie $f_v(a, \ra)$.
Consider the stream of edges $e: (a, *; t)$, \ie edges from node $a$ to any other node indicated by a wildcard $*$.
Drop the destination (\ie the wildcard) of each edge.  The stream now becomes a stream of tuples $(a, \omega(t))$ where $\omega(t)$ is the edge weight from node $a$ at time $t$.  When a query is posed to find the out-degree of node $a$, \cm~\cite{DBLP:journals/jal/CormodeM05} returns the minimum of the weights in different hash buckets as the estimation of the flow out of node $a$.  

If the number of unique neighbors (\ie connected using one hop ignoring the weights) is sought, we adapt the 
procedure above by replacing $\omega(t)$ with 1 in the stream above for all nodes.  Clearly, by construction, the answer obtained is an over-estimation of the true out-degree because of two reasons: (a) collisions in the hash-buckets, and (b) self-collision, \ie we do not know if an edge has been seen previously and thus count the outgoing edge again even if we have seen it. 

The case for in-degree point queries is similar.
To compute the in-degree of a node, we simply adapt the stream to create a one-dimensional stream as in the case above.  Drop the source for each edge to create a stream.
The variations in the case of out-degree for the total in-flow and the number of neighbors with in-links to a node can be estimated as in the case of out-degree outlined above.  

The error estimates for point queries (see Sec.~4.1 of \cite{DBLP:journals/jal/CormodeM05}) hold for these cases.

\begin{lemma}
\label{degree}
The estimated out-degree (in-degree) is within a factor of $\epsilon$ with probability $\delta$ if we use $d=\lceil ln(1/\delta) \rceil$ rows of pair-wise independent hash functions and $w=\lceil e/\epsilon \rceil$.
\end{lemma}

\eat{

\stitle{Path.}
$(a, b), (b, c)$.

\stitle{Star.}
$(a, b), (a, c), (a, d)$.

\stitle{Cycle.}
$(a, b), (b, c), (c, a)$.

\begin{figure}[t]
\hspace*{2ex}
\begin{minipage}{0.25\columnwidth}
\centerline{
\xymatrixcolsep{0.1in}
\xymatrix{
a \ar[d] \\
b \\
} 
} 
\centerline{
\xymatrixrowsep{0.2in}
\xymatrix{
& \txt{(a) Edge} &
} 
} 
\end{minipage}
\hspace*{-2ex}
\begin{minipage}{0.25\columnwidth}
\centerline{
\xymatrixcolsep{0.1in}
\xymatrix{
a \ar[d] & \\
b \ar[r] & c \\
} 
} 
\centerline{
\xymatrixrowsep{0.2in}
\xymatrix{
& \txt{(b) Path} &
} 
} 
\end{minipage}
\hspace*{-2ex}
\begin{minipage}{0.25\columnwidth}
\centerline{
\xymatrixcolsep{0.1in}
\xymatrix{
& a \ar[d]\ar[dl]\ar[dr] & \\
b & c & d \\
} 
} 
\centerline{
\xymatrixrowsep{0.2in}
\xymatrix{
& \txt{(c) Star} &
} 
} 
\end{minipage}
\hspace*{-2ex}
\begin{minipage}{0.25\columnwidth}
\centerline{
\xymatrixcolsep{0.1in}
\xymatrix{
a \ar[d] & \\
b \ar[r] & c \ar[ul] \\
} 
} 
\centerline{
\xymatrixrowsep{0.2in}
\xymatrix{
& \txt{(d) Circle} &
} 
} 
\end{minipage}
\caption{Studied cases}
\label{fig:2hash}
\end{figure}

}

\section{Implementation Details}
\label{sec-implementation}

In this section, we first discuss the data structures used to implement \sys (Section~\ref{subsec-adjacent}).
We then introduce the definition of pairwise independent hash functions (Section~\ref{subsec-pairwise}).
We also discuss its potential extension in a distributed environment (Section~\ref{subsec:distributed}).

\subsection{Adjacent Matrix and Its Extension}
\label{subsec-adjacent}

Using hash-based methods, although it is evident that it only requires 1-pass of the graph stream to construct/update a \sys.
However, the linear time complexity of constructing and updating a \sys depends on the data structures used. 
For example, the adjacent list may not be a fit, since searching a specific edge is not in constant time.

\subsubsection{Adjacent Matrix}
\label{subsubsec-matrix}

An adjacency matrix is a means of representing which nodes of a graph are adjacent to which other vertices.

\begin{example}
\label{exam:matrix}
Consider the sketch $S_1$ in Fig.~\ref{fig:2hash}~(a) for example. Its adjacent matrix is shown in Fig.~\ref{fig:matrix}.
\end{example}

Example~\ref{exam:matrix} showcases a directed graph.
In the case of an undirected graph, it will be a symmetric matrix.

\begin{figure}[t]
\centering
\begin{tabular}{|c|c|c|c|c|}
	\multicolumn{1}{c}{{\color{white} {\Huge j}} $\at{from}\setminus\at{to}$} & \multicolumn{1}{c}{$I(af)$} & \multicolumn{1}{c}{$II(bc)$} & \multicolumn{1}{c}{$III(dg)$}  & \multicolumn{1}{c}{$IV(e)$}\\ \cline{2-5}
	\multicolumn{1}{c|}{$I(af)$} & 1 & 2 & 0 & 0 \\ \cline{2-5}
	\multicolumn{1}{c|}{$II(bc)$} & 3 & 1 & 1 & 1 \\ \cline{2-5}
	\multicolumn{1}{c|}{$III(dg)$} & 0 & 1 & 1 & 0 \\ \cline{2-5}
	\multicolumn{1}{c|}{$IV(e)$} & 1 & 1 & 1 & 0 \\ \cline{2-5}
\end{tabular}
\caption{The adjacent matrix of $S_1$}
\label{fig:matrix}
\end{figure}

\stitle{Construction.}
Consider a graph stream $G=\langle e_1, e_2, $ $\cdots, e_m \rangle$ where $e_i = (x_i, y_i; t_i)$. Given a number of nodes $w$ and a hash function $h(\cdot) \ra [1, w]$. We use the following strategy.

\etitle{Step 1.} [Initialization.]
Construct a $w\times w$ matrix ${\cal M}$, with all values initialized to be $0$.

\etitle{Step 2.} [Insertion of $e_i$.]
Compute $h(x_i)$ and $h(y_i)$. Increase the value of ${\cal M}[h(x_i)][h(y_i)]$ by $\omega(e_i)$, the weight of $e_i$.

In the above strategy, step 1 will take constant time to allocate a matrix. Step 2 will take constant time for each $e_i$. Hence, the time complexity is $O(m)$ where $m$ is the number of edges in $G$. The space used is $O(w^2)$.

\stitle{Deletions.} [Deletion of $e_i$.]
Insertions have been discussed in the above step 2. 
For the deletion of an $e_i$ that is not of interest (\eg out of a certain time window), it is simply to decrease the value of ${\cal M}[h(x_i)][h(y_i)]$ by $\omega(e_i)$ in $O(1)$ time.

Alternatively, one may consider to use an adjacent hash-list is to maintain, for each vertex, a list of its adjacent nodes using a hash table. Given an edge $e_i (x_i, y_i; t_i)$, two hash operations are needed: The first is to locate $x_i$, and the second is to find $y_i$ from $x_i$'s hash-list. Afterwards, it updates the corresponding edge weight.
Adjacent list is known to be suitable when graph is sparse. However, in terms of compressed graph in our case, as will be shown later in experiments, most sketches are relatively dense, which makes the adjacent matrix the default data structure to manage our graph sketches.

\subsubsection{Using Non-Square Matrices}
\label{subsubsec-matrix2}

When using a classical square matrix for storing a sketch, we have an $n * n$ matrix. 
Consider all edges from node $a$ such as $(a, *)$. 
Using any hash function will inevitably hash all of these edges to the same row.
For example, in Fig.~\ref{fig:matrix}, all edges $(a, *)$ will be hashed to the first row of the matrix.

When there is only one hash function to use, due to the lack of {\em a priori} knowledge of data distribution, it is hard to decide the right shape of a matrix. However, the application of $d$ sketches provides us an opportunity to heuristically reduce the chance of combined hash collisions.

The basic idea is, instead of using an $n * n$ matrix with one hash function, we use an $m * p$ matrix with two hash functions: $h_1(\cdot) \ra [1, m]$ on the {\em from} nodes and $h_2(\cdot) \ra [1, p]$ on the {\em to} nodes.

\begin{figure}[t]
\centering
\begin{tabular}{|c|c|c|}
	\multicolumn{1}{c}{{\color{white} {\Huge j}} $\at{from}\setminus\at{to}$} & \multicolumn{1}{c}{$i(abcd)$} & \multicolumn{1}{c}{$ii(efg)$} \\ \cline{2-3}
	\multicolumn{1}{c|}{$I(a)$} & 2 & 0  \\ \cline{2-3}
	\multicolumn{1}{c|}{$II(b)$} & 3 & 1 \\ \cline{2-3}
	\multicolumn{1}{c|}{$III(c)$} & 0 & 2  \\ \cline{2-3}
	\multicolumn{1}{c|}{$IV(d)$} & 0 & 1 \\ \cline{2-3}
	\multicolumn{1}{c|}{$IV(e)$} & 2 & 1  \\ \cline{2-3}
	\multicolumn{1}{c|}{$IV(f)$} & 1 & 0  \\ \cline{2-3}
	\multicolumn{1}{c|}{$IV(g)$} & 1 & 0  \\ \cline{2-3}			
\end{tabular}
\caption{A non-square matrix}
\label{fig:matrix2}
\end{figure}

\begin{example}
\label{exam:matrix2}
Consider the graph stream in Fig.\ref{fig:graph}.
Assume that we use two hash functions:  $h_1(\cdot) \ra [1, 7]$ and  $h_2(\cdot) \ra [1, 2]$. 
The non-square matrix is shown in Fig.~\ref{fig:matrix2}.
\end{example}

In practice, when we can generate multiple sketches, we heuristics use matrices $n * n$, $2n * n/2$, $n/2 * 2n$, $4n * n/4$, $n/4 * 4n$, etc. That is, we use matrices with the same sizes but different shapes.

\subsection{Pairwise Independent Hash Functions}
\label{subsec-pairwise}

Here, we only borrow the definition of pairwise independent hash functions\footnote{\url{http://people.csail.mit.edu/ronitt/COURSE/S12/handouts/lec5.pdf}}, while relying on existing tools to implement them.

A family of functions ${\cal H} = \{h | h(\cdot) \ra [1, w]\}$ is called a {\em family of pairwise independent hash functions} if for two different hash keys $x_i$, $x_j$, and $k, l \in [1, w]$, 
\[ 
	\at{Pr}_{h \la {\cal H}}[h(x_i) = k \wedge h(x_j) = l] =  1 / w^2
\]

Intuitively, when using multiple hash functions to construct sketches, the hash functions used should be pairwise independent in order to reduce the probability of hash collisions. Please refer to~\cite{DBLP:journals/jal/CormodeM05} for more details.

\subsection{Discussion: Distributed Environment}
\label{subsec:distributed}

In this paper, we mainly discuss how to implement \sys in a centralized environment.
However, it is easy to observe that \sys can be easily applied to a distributed environment, since the construction and maintenance of each sketch is independent of each other. Assuming that we have $d$ sketches in one computing node, when $m$ nodes are available, we can use $d\times m$ pairwise independent hash functions, which may significantly reduce the probability of errors.

\section{Conclusion}
\label{sec-conclusion}

We have proposed a new graph sketch for summarizing graph streams.
We have demonstrated its wide applicability to many emerging applications.

{
\bibliographystyle{abbrv}
\bibliography{DA}
}

\end{document}